\documentclass[a4paper, 10pt, conference]{ieeeconf}%
\usepackage{amsmath}
\usepackage{amsfonts}
\usepackage{amssymb}
\usepackage{graphicx}
\usepackage{cite}%
\IEEEoverridecommandlockouts
\newtheorem{theorem}{Theorem}

\newtheorem{definition}{Definition}
\newtheorem{example}{Example}

\newtheorem{proposition}{Proposition}
\newtheorem{remark}{Remark}

\begin{document}

\title{An Information-theoretic Approach to Privacy}
\author{Lalitha Sankar, S. Raj Rajagopalan, and H. Vincent Poor \thanks{This research
is supported in part by the U. S. National Science Foundation under Grant
CCF-1016671, the Air Force Office of Scientific\ Research under Grant
FA9550-09-1-0643, and by a fellowship from the Princeton University Council
on\ Science and Technology. }\thanks{L. Sankar and H. V. Poor are with the
Department of Electrical Engineering, Princeton University, Princeton, NJ,
USA. S. Raj Rajagopalan is with HP\ Labs, Princeton, NJ, USA.
\{lalitha,poor@princeton.edu,raj.rajagopalan@hp.com\texttt{{\small \}}}}}
\maketitle

\begin{abstract}
Ensuring the usefulness of electronic data sources while providing necessary
privacy guarantees is an important unsolved problem. This problem drives the
need for an overarching analytical framework that can quantify the safety of
personally identifiable information (privacy) while still providing a
quantifable benefit (utility) to multiple legitimate information consumers.
State of the art approaches have predominantly focused on privacy. This paper
presents the first information-theoretic approach that promises an analytical
model guaranteeing tight bounds of how much utility is possible for a given
level of privacy and vice-versa.

\end{abstract}

\thispagestyle{empty} \pagestyle{empty}



\section{The Database Privacy\ Problem}

Information technology and electronic communications have been rapidly applied
to almost every sphere of human activity, including commerce, medicine and
social networking. The concomitant emergence of myriad large centralized
searchable data repositories has made \textquotedblleft
leakage\textquotedblright\ of private information such as medical data, credit
card information, or social security numbers via data correlation
(inadvertently or by malicious design) highly probable and thus an important
and urgent societal problem. Unlike the well-studied secrecy problem (e.g.,
\cite{RSA,CsisNar1,Wyner}) in which the protocols or primitives make a sharp
distinction between secret and non-secret data, in the \emph{privacy} problem,
disclosing data provides informational utility while enabling possible loss of
privacy at the same time. In fact, in the course of a legitimate transaction,
a user learns some public information, which is allowed and needs to be
supported for the transaction to be meaningful, but at the same time he can
also learn/infer private information, which needs to be prevented. Thus every
user is (potentially)\ also an adversary. This drives the need for a unified
analytical framework that can tell us unequivocally and precisely how safe
private data can be (privacy) and simultaneously provide measurable benefit
(utility) to multiple legitimate information consumers.

It has been noted that utility and privacy are competing goals:
\textit{perfect privacy can be achieved by publishing nothing at all, but this
has no utility; perfect utility can be obtained by publishing the data exactly
as received, but this offers no privacy }\cite{Chawla01}. Utility of a data
source is potentially (but not necessarily) degraded when it is restricted or
modified to uphold privacy requirements. The central problem of this paper is
a precise quantification of the tradeoff between the privacy needs of the
\textit{respondents} (individuals represented by the data)\ and the utility of
the \textit{sanitized} (published) data for any data source.

Though the problem of privacy and information leakage has been studied for
several decades by multiple research communities (e.g.,
\cite{Chawla01,Adam_Wort,Dalenius,Sweeney,Ag_Ag} and the references therein),
the proposed solutions have been both heuristic and application-specific. The
recent groundbreaking theory of $\epsilon$-differential privacy
\cite{Dwork_DP} from the theoretical computer science community provides the
first universal metric of privacy that applies to any numerical database. We
seek to address the open question of a universal and analytical
characterization that provides a tight privacy-utility tradeoff using tools
and techniques from information theory.

Rate distortion theory is a natural choice to study the utility-privacy
tradeoff; utility can be quantified via fidelity which, in turn, is related to
\textit{distortion}, and privacy can be quantified via \textit{equivocation}.
Our key insight is captured in the following theorem which is presented in
this paper: for a data source with private and public data, \textit{minimizing
the information disclosure rate }sufficiently to satisfy the desired utility
for the public data \textit{is equivalent to maximizing the privacy} for the
private data. In a sparsely referenced paper \cite{Yamamoto} from three
decades ago, Yamamoto developed the tradeoff between rate, distortion, and
equivocation for a specific and simple source model. In this paper, we show
via the above summarized theorem that Yamamoto's formalism can be translated
into the language of data disclosure. Furthermore, we develop a framework that
allows us to model data sources, specifically databases, develop application
independent utility and privacy metrics, quantify the fundamental bounds on
the utility-privacy tradeoffs, and develop a side-information model for
dealing with questions of external knowledge.

The paper is organized as follows. We present channel model and preliminaries
in Section \ref{Sec_II}. The main result and the proof are developed in
Section \ref{Sec_III}. We discuss the results and present numerical examples
in\ Section \ref{Sec_IV}. We conclude in Section \ref{Sec_V}.

\section{\label{Sec_II}The Database Privacy Problem}

\subsection{Problem Definition}

While the problem of quantifying the utility/privacy problem applies to all
types of data sources, we start our study with databases because they are
highly structured and historically better studied than other types of sources.
A database is a table (matrix) whose rows represent individual entries and
whose columns represent the \emph{attributes} of each entry \cite{Adam_Wort}.
For example, the attributes of each entry in a healthcare database typically
include name, address, social security number (SSN), gender, and a collection
of medical information, and each entry contains the information pertaining to
an individual. Messages from a \emph{user} to a database are called
\emph{queries} and, in general, result in some numeric or non-numeric
information from the database termed the \emph{response}.

The goal of privacy protection is to ensure that, to the extent possible, the
user's knowledge is not increased beyond strict predefined limits by
interacting with the database. The goal of utility provision is, generally, to
maximize the amount of information that the user can receive. Depending on the
relationships between attributes, and the distribution of the actual data, a
response may contain information that can be inferred beyond what is
explicitly included in the response. The privacy policy defines the
information that should not be revealed explicitly or by inference to the user
and depends on the context and the application. For example, in a database on
health statistics, attributes such as name and SSN may be considered private
data, whereas in a state motor vehicles\ database only the SSN\ is considered
private. The challenge for privacy protection is to design databases such that
responses do not reveal information contravening the privacy policy.

\subsection{Current Approaches and Metrics}

The problem of privacy in databases has a long and rich history stretching
back to the 1970s and space restrictions preclude any attempt to do full
justice to the different approaches that have been considered along the way.
While there have been many heuristic approaches to privacy, we only present
the major milestones in privacy research on creating quantitative privacy
metrics. Since privacy is a requirement that appears in many diverse contexts,
a robust and formal notion of privacy that satisifies most, if not all,
requirements is a tricky proposition and there have been many attempts at a
definition. The reader is referred to the excellent survey by Dwork
\cite{Dwork_ACM} for a detailed history of the field. The problem of privacy
was first exposed by census statisticians who were required to publish
statistics related to census functions but without revealing any particulars
of individuals in the census databases. An early work by Dalenius
\cite{Dalenius} reveals the depth to which this problem was considered.
Several early attempts were made to publish census data using ad hoc
techniques such as sub-sampling. However, the first widely reported attempt at
a formal definition of privacy was by Sweeney \cite{Sweeney}. The concept of
$k$-\textit{anonymity} proposed by Sweeney captures the intuitive notion of
privacy that every individual entry should be indistinguishable from $(k-1)$
other entries for some large value of $k$. This notion of anonymity for
database respondents is analogous to similar proposals that were made for
anonymity on the Internet such as crowds \cite{AviRubin_Crowds}. More
recently, researchers in the data mining community have proposed to quantify
the privacy loss resulting from data disclosure as the mutual information
between attribute values in the original and perturbed data sets, both modeled
as random variables \cite{Ag_Ag}.

The approaches considered in the literature have centered on the correct
application of \emph{perturbation} (also called \textit{sanitization}), which
encompasses a general class of database modification techniques that ensure
that a user interacts only with a modified database that is derived from the
original (e.g.: \cite{Dalenius,Sweeney,Ag_Ag,Chawla01}). Most of the these
perturbation approaches, with the exception of differential privacy-based
ones, are heuristic and application-specific and often focus on additive noise approaches.

\textit{Differential privacy: }More recently, privacy approaches for
statistical databases has been driven by the differential privacy definition
\cite{Dwork_DP,DS_DP,Dwork_DP_Survey,Dwo_NoiSen}. In these papers, the authors
take the view that privacy of an individual in a database is related to the
ability of an adversary to detect whether that individual's data is in that
database or not. Motivated by cryptographic models, they formalize this
intuition by defining the difference in the adversary's outputs when presented
with two databases $D$ and $D^{^{\prime}}$that are identical except in one row.

\begin{definition}
[\cite{Dwork_ACM}]\label{Def_DP}A function $\mathcal{K}$ gives $\epsilon
$-differential privacy if for all databases $D,D^{^{\prime}}$ defined as
above, and all $S\subseteq$ Range($\mathcal{K}$),
\begin{equation}
\Pr[\mathcal{K}(D)\in S]\leq exp(\epsilon)\cdot\Pr\left[  \mathcal{K}\left(
D^{^{\prime}}\right)  \in S\right]
\end{equation}
where the probability space in each case is over the coin flips of
$\mathcal{K}$.
\end{definition}

It is important to make two observations regarding the above definition.
First, the probabilities in definition \ref{Def_DP} are over the actions of
the function $\mathcal{K}$ and not over the distribution of $D$; in other
words, the definition is independent of the distribution from which $D$ may be
sampled. Second, Definition \ref{Def_DP} guarantees that the presence or
absence of an individual row in the database makes very little difference to
the output of the adversary as required, and thus, provides a precise privacy
guarantee to any individual in the database.

More recently, Dwork \textit{et al}. \cite{Dwo_NoiSen} also provide a
mechanism for achieving $\epsilon$-differential privacy universally for
statistical queries (queries that map subsets of database entries to real
numbers) which we summarize below. Let $Z\sim Lap(b)$ represent a Laplacian
distributed random variable with parameter $b.$ If $b=1/\epsilon$ we have that
the density at $z$ is proportional to $\exp(-4|z|)$ and for any $(z,z^{\prime
})$ such that $|z-z^{\prime}|\leq1$, $\Pr(z)$ and $\Pr(z^{^{\prime}})$ are
within a factor of $e^{\epsilon}$. The following proposition shows that it is
possible to achieve $\epsilon$-differential privacy for a given statistical
query class for suitable choice of the Laplacian parameter.

\begin{proposition}
[\cite{Dwo_NoiSen} ]\label{Dw_LMech}For any statistical query $f:D\rightarrow
\mathcal{R}$, the mechanism $L$ that adds independently generated noise to the
output terms with distribution $Lap(\left.  \Delta_{f}\right/  \epsilon$)
guarantees $\epsilon$-differential privacy where $\Delta_{f}=\max\left\vert
f(D)-f(D^{^{\prime}})\right\vert $ for $D$, $D^{^{\prime}}$ which are
different in exactly one row.
\end{proposition}

Proposition \ref{Dw_LMech} is the most significant milestone in the theory of
privacy because it provides a method to guarantee a strong but quantifiable
notion of privacy for statistical databases independent of their content.
Furthermore, the noise distribution can be chosen after seeing the query, so
that the noise level can be adjusted adaptively when presented with a sequence
of queries. However, one constraint in using Proposition \ref{Dw_LMech} to
define $\epsilon$ is that $\Delta f$ may be difficult to estimate -- a loose
bound on $\Delta f$ \ may result in an overly large $\epsilon$, thereby
resulting in a possible degradation of utility.

To date, privacy has been the main focus of most work in this area. Indeed,
Dwork \cite{Dwork_DP} says explicitly that privacy is paramount in their work.
However, databases exist to be useful and implementing sanitization techniques
may hurt the usefulness of the database while safeguarding privacy. In much of
the earlier work on database privacy, the utility is implicit. For exmple,
Sweeney assumes that the databases can be $k$-anonymized and still maintain
usefulness. However, without a relationship between $k$ and some formal notion
of usefulness, it is impossible to say what a reasonable value of $k$\ should
be in reality. Similarly, utility in privacy-preserving techniques such as
clustering \cite{Chawla01} and histograms \cite{Chawla02} is assumed to be
guaranteed as a direct result of the methods used; for example, in
\cite{Chawla02} it is shown that approximation algorithms that can run on
original histograms can also run on the sanitized histograms with a
degradation of performance. Clustering, a common sanitization technique
\cite{Chawla01,Motwani1,L-diversity}, is claimed to maintain utility as a
result of the following property: all points in a cluster are mapped to the
cluster center, so no point is moved more than the diameter of the largest cluster.

The differential privacy model uses additive noise for sanitization which in
turn suggests a utility metric related to the accuracy of the sanitized
database. The Laplacian noise model was chosen for achieving differential
privacy in part because the mean and mode are zero, in which case no noise is
added in most cases. The privacy parameter $\epsilon$\ is inversely related to
the variance of the added noise -- a better privacy guarantee requires a
smaller $\epsilon$\ which in turn implies higher variance. The accuracy of a
sanitized database as a whole is inversely related to the privacy requirement.
Determining the appropriate range of $\epsilon$\ so that both privacy and
accuracy requirements are balanced requires knowledge of the specific
application. As an example, in the case of learning, recent results
\cite{Sarwate_C} in the area of \textit{private learning} bound the extent to
which the performance (i.e. accuracy)\ of certain kinds of classifers degrade
when the training data is sanitized using the $L$ mechanism in Proposition
\ref{Dw_LMech}. In such cases, it is possible to have both, differential
privacy with a known $\epsilon,$\ as well as quantified utility loss for the
application under consideration.

\subsection{Privacy vs. Secrecy}

It is important to contrast the privacy problem from the well-studied
(cryptographic and information-theoretic)\ \emph{secrecy} problem where the
task is to stop specific information from being received by untrusted third
parties (eavesdroppers, wire-tappers, and other kinds of adversaries). In the
\textit{private information retrieval} model \cite{Gasarch04asurvey}, the
privacy problem is inverted in that the adversary is the database from whom
the user wants to keep his \textit{query} secret. In the secure multi-party
computation model \cite{multiparty}, each player wishes to keep his
\textit{entire input} secret from the other players while jointly computing a
function on all the inputs. In all these problems, a specific data item is
clearly either secret or public, whereas in the privacy problem, the same data
while providing informational utility to the user can reveal private
information about the individuals represented by the data. This eliminates the
possibility of using secrecy techniques such as a specific model of the
adversary or of harnessing any computing \cite{Goldreich1} or physical
advantages such as secret keys, channel differences, or side information
\cite{Poor9}.

\section{\label{Sec_III}An Information-Theoretic Approach}

\subsection{\label{Sec_DB_Model}Model for Databases}

\textit{Circumventing the semantic issue}: In general, utility and privacy
metrics tend to be application specific. Focusing our efforts on developing an
analytical model, we propose to capture a canonical database model and
representative abstract metrics. Such a model will circumvent the classic
privacy issues related to the semantics of the data by assuming that there
exist forward and reverse maps of the data set to the proposed abstract format
(for e.g., a string of bits or a sequence of real values). Such mappings are
often implicitly assumed in the privacy literature
\cite{Chawla01,Ag_Ag,Dwork_DP}; our motivation for making it explicit is to
separate the semantic issues from the abstraction and apply Shannon-theoretic techniques.

\textit{Model}: Our proposed model focuses on large databases with $K$
attributes per entry. Let $X_{k}\in\mathcal{X}_{k}$ be a random variable
denoting the $k^{th}$ attribute, $k=1,2,\ldots,K,$ and let $\mathbf{X}%
\equiv\left(  X_{1},X_{2},\ldots,X_{K}\right)  $. A database $d$ with $n$ rows
is a sequence of $n$ independent observations of $\mathbf{X}$ from the
distribution%
\begin{equation}
p_{\mathbf{X}}\left(  \mathbf{x}\right)  =p_{X_{1}X_{2}\ldots X_{K}}\left(
x_{1},x_{2},\ldots,x_{K}\right)  \label{Prob_JointDist}%
\end{equation}
which is assumed to be known to both the designers and users of the database.
Our simplifying assumption of row independence holds generally (but not
always) as correlation is typically across attributes and not across entries.
We write $\mathbf{X}^{n}=\left(  X_{1}^{n},X_{2}^{n},\ldots,X_{K}^{n}\right)
$ to denote the $n$ independent observations of $\mathbf{X}$. This database
model is universal in the sense that most practical databases can be mapped to
this model.

A joint distribution in (\ref{Prob_JointDist}) models the fact that the
attributes in general are correlated and can reveal information about one
another. In addition to the revealed information, a user of a database can
have access to correlated side information from other information sources. We
model the side-information as an $n$-length sequence $Z^{n}$ which is
correlated with the database entries via a joint distribution $p_{\mathbf{X}%
Z}\left(  \mathbf{x,}z\right)  .$

\textit{Public and private variables}: We consider a general model in which
some attributes need to be kept private while the source can reveal a function
of some or all of the attributes. We write $\mathcal{K}_{r}$ and
$\mathcal{K}_{h}$ to denote sets of private (subscript $h$ for hidden) and
public (subscript $r$ for revealed) attributes, respectively, such that
$\mathcal{K}_{r}\cup\mathcal{K}_{h}=\mathcal{K\equiv}\left\{  1,2,\ldots
,K\right\}  $. We further denote the corresponding collections of public and
private attributes by $\mathbf{X}_{r}\equiv\left\{  X_{k}\right\}
_{k\in\mathcal{K}_{r}}$ and $\mathbf{X}_{h}\equiv\left\{  X_{k}\right\}
_{k\in\mathcal{K}_{h}}$, respectively. Our notation allows for an attribute to
be both public and private; this is to account for the fact that a database
may need to reveal a function of an attribute while keeping the attribute
itself private. In general, a database can choose to keep public (or private)
one or more attributes ($K>1)$. Irrespective of the number of private
attributes, a non-zero utility results only when the database reveals an
appropriate function of some or all of its attributes.

\textit{Special cases}: For $K=1$, the lone attribute of each entry (row) is
both public and private, and thus, we have $X\equiv X_{r}\equiv X_{h}$. Such a
model is appropriate for data mining \cite{Ag_Ag}; for a more general case in
which $K_{h}=K_{r}=K$, we obtain a model for census \cite{Dalenius,Chawla01}
data sets in which utility generally is achieved by revealing a function of
every entry of the database while simultaneously ensuring that no entry is
perfectly revealed. For $K=2$ and $\mathcal{K}_{h}\cup\mathcal{K}%
_{r}=\mathcal{K}$ and $\mathcal{K}_{h}\cap\mathcal{K}_{r}=\emptyset,$ we
obtain the Yamamoto model in \cite{Yamamoto}.

\subsection{Metrics: The Privacy and Utility Principle}

Even though utility and privacy measures tend to be specific to the
application, there is a fundamental principle that unifies all these measures
in the abstract domain. The aim of a privacy-preserving database is to provide
some measure of utility to the user while at the same time guaranteeing a
measure of privacy for the entries in the database.

A user perceives the utility of a perturbed database to be high as long as the
response is similar to the response of the unperturbed database; thus, the
utility is highest of an unperturbed database and goes to zero when the
perturbed database is completely unrelated to the original database.
Accordingly, our utility metric is an appropriately chosen average `distance'
function between the original and the perturbed databases. Privacy, on the
other hand, is maximized when the perturbed response is completely independent
of the data. Our privacy metric measures the difficulty of extracting any
private information from the response, i.e., the amount of uncertainty or
\textit{equivocation }about the private attributes given the response.

\subsection{\label{SS2}Utility-Privacy Tradeoffs}

\subsubsection{A Privacy-Utility Tradeoff Model}

We now propose a privacy-utility model for databases. \textit{Our primary
contribution is demonstrating the equivalence between the database privacy
problem and a source coding problem with additional privacy constraints}. A
primary motivation for our approach is the observation that database
sanitization is traditionally the process of distorting the data to achieve
some measure of privacy. For our abstract universal database model,
sanitization is thus a problem of mapping a set of database entries to a
different set subject to specific utility and privacy requirements.

Our notation below relies on this abstraction.\ Recall that a database $d$
with $n$ rows is an instantiation of $\mathbf{X}^{n}$. Thus, we will
henceforth refer to a real database $d$ as an \textit{input sequence} and to
the corresponding sanitized database (SDB) $d^{\prime}$ as an \textit{output
sequence}. When the user has access to side information, the
\textit{reconstructed sequence} at the user will in general be different from
the SDB sequence.

Our coding scheme consists of an encoder $F_{E}$ which is a mapping from the
set of all input sequences (i.e., all databases $d$ chosen from an underlying
distribution$)$ to a set of indices~$\mathcal{W}\equiv\left\{  1,2,\ldots
,M\right\}  $ and an associated table of output sequences (each of which is a
$d^{\prime})$ with a one-to-one mapping to the set of indices given by%
\begin{equation}
F_{E}:\left(  \mathcal{X}_{1}^{n}\times\mathcal{X}_{2}^{n}\times\ldots
\times\mathcal{X}_{k}^{n}\right)  _{k\in\mathcal{K}_{enc}}\rightarrow
\mathcal{W}\equiv\left\{  SDB_{k}\right\}  _{k=1}^{M} \label{F_Enc}%
\end{equation}
where $\mathcal{K}_{r}\subseteq\mathcal{K}_{enc}\subseteq\mathcal{K}$ and
$M=2^{nR}$ is the number of output (sanitized) sequences created from the set
of all input sequences. The encoding rate $R$ is the number of bits per row
(without loss of generality, we assume $n$ rows in $d$ and $d^{\prime}$) of
the sanitized database. The encoding $F_{E}$ in (\ref{F_Enc}) includes both
public and private attributes in order to model the general case in which the
sanitization depends on a subset of all attributes.

A user with a view of the SDB (i.e., an index $w\in\mathcal{W}$ for every $d)$
and with access to side information $Z^{n}$, whose entries $Z_{i}$,
$i=1,2,\ldots,n,$ take values in the alphabet $\mathcal{Z}$, reconstructs the
database $d^{\prime}$ via the mapping%
\begin{equation}
F_{D}:\mathcal{W}\times\mathcal{Z}^{n}\rightarrow\left\{  \mathbf{\hat{x}%
}_{r,m}^{n}\right\}  _{m=1}^{M}\in\left(
{\textstyle\prod\nolimits_{k\in\mathcal{K}_{r}}}
\mathcal{\hat{X}}_{k}^{n}\right)  \label{F_Dec}%
\end{equation}
where $\mathbf{\hat{X}}_{r}^{n}=F_{D}\left(  F_{E}\left(  \mathbf{X}%
^{n}\right)  \right)  $.

A database may need to satisfy multiple utility constraints for different
(disjoint) subsets of attributes, and thus, we consider a general framework
with $L\geq1$ utility functions that need to be satisfied. Relying on the
distance based utility principle, we model the $l^{th}$ utility,
$l=1,2,\ldots,L,$ via the requirement that the average \textit{distortion}
$\Delta_{l}$ of the revealed variables is upper bounded, for some $\epsilon
>0$, as%
\begin{multline}
u_{l}:\Delta_{l}\equiv\mathbb{E}\left[  \frac{1}{n}%
{\textstyle\sum_{i=1}^{n}}
g\left(  \mathbf{X}_{r,i},\mathbf{\hat{X}}_{r,i}\right)  \right]  \leq
D_{l}+\epsilon\text{, }\label{Utility_mod}\\
l=1,2,\ldots,L,
\end{multline}
where $g\left(  \cdot,\cdot\right)  $ denotes a distortion function,
$\mathbb{E}$ is the expectation over the joint distribution of $(\mathbf{X}%
_{r},\mathbf{\hat{X}}_{r})$, and the subscript $i$ in $\mathbf{X}_{r,i}$ and
$\mathbf{\hat{X}}_{r,i}$ denotes the $i^{th}$ entry of $\mathbf{X}_{r}^{n}$
and $\mathbf{\hat{X}}_{r}^{n}$, respectively. Examples of distance-based
distortion functions include the Euclidean distance for Gaussian distributed
database entries, the Hamming distance for binary input and output sequences,
and the Kullback-Leibler (K-L) `distance' comparing the input and output distributions.

Having argued that a quantifiable uncertainty captures the underlying privacy
principle of a database, we model the uncertainty or equivocation about the
private variables using the entropy function as
\begin{equation}
p:\Delta_{p}\equiv\frac{1}{n}H\left(  \mathbf{X}_{h}^{n}|W,Z^{n}\right)  \geq
E-\epsilon, \label{Equivoc}%
\end{equation}
i.e., we require the average number of uncertain bits per entry to be lower
bounded by $E$. The case in which side information is not available at the
user is obtained by simply setting $Z^{n}=0$ in (\ref{F_Dec}) and
(\ref{Equivoc}).

The utility and privacy metrics in (\ref{Utility_mod}) and (\ref{Equivoc}),
respectively, capture two aspects of our universal model: a) both represent
averages by computing the metrics across all database instantiations $d$, and
b) the metrics bound the average distortion and privacy per entry. Thus, as
the likelihood of the non-typical sequences decreases exponentially with
increasing $n$ (very large databases), these guarantees apply nearly uniformly
to all (typical)\ entries. Our general model also encompasses the fact that
the exact mapping from the distortion and equivocation domains to the utility
and privacy domains, respectively, can depend on the application domain. We
write $D\equiv(D_{1},D_{2},\ldots,$ $D_{L})$ and $\Delta\equiv(\Delta
_{1},\Delta_{2},\ldots,\Delta_{L})$. Based on our notation thus far, we define
the utility-privacy tradeoff region as follows.

\begin{definition}
The utility-privacy tradeoff region $\mathcal{T}$ is the set of all feasible
utility-privacy tuples $(D,E)$ for which there exists a coding scheme $\left(
F_{E},F_{D}\right)  $ given by (\ref{F_Enc}) and (\ref{F_Dec}), respectively,
with parameters $(n,M,\Delta,\Delta_{p})$ satisfying the constraints in
(\ref{Utility_mod}) and (\ref{Equivoc}).
\end{definition}

\subsubsection{\label{Sec_RDE}Equivalence of Utility-Privacy and
Rate-Distortion-Equivocation}

We now present an argument for the equivalence of the above utility-privacy
tradeoff analysis with a rate-distortion-equivocation analysis of the same
source. For the database source model described here, a classic lossy source
coding problem is defined as follows.

\begin{definition}
The set of tuples $(R,D)$ is said to be feasible (achievable) if there exists
a coding scheme given by (\ref{F_Enc}) and (\ref{F_Dec}) with parameters
$(n,M,\Delta)$ satisfying the constraints in (\ref{Utility_mod}) and a rate
constraint
\begin{equation}
M\leq2^{n\left(  R+\epsilon\right)  }. \label{Rate_constraint}%
\end{equation}

\end{definition}

When an additional privacy constraint in (\ref{Equivoc}) is included, the
source coding problem becomes one of determining the achievable
rate-distortion-equivocation region defined as follows.

\begin{definition}
\label{Def_RDE}The rate-distortion-equivocation region $\mathcal{R}$ is the
set of all tuples $(R,D,E)$ for which there exists a coding scheme given by
(\ref{F_Enc}) and (\ref{F_Dec}) with parameters $(n,M,\Delta,\Delta_{p})$
satisfying the constraints in (\ref{Utility_mod}), (\ref{Equivoc}), and
(\ref{Rate_constraint}). The set of all feasible distortion-equivocation
tuples $\left(  D,E\right)  $ is denoted by $\mathcal{R}_{D-E}$, the
equivocation-distortion function in the $D$-$E$ plane is denoted by
$\Gamma(D)$, and the distortion-equivocation function which quantifies the
rate as a function of both $D$ and $E$ is denoted by $R\left(  D,E\right)  $.
\end{definition}

Thus, a rate-distortion-equivocation code is by definition a (lossy) source
code satisfying a set of distortion constraints that achieves a specific
privacy level for every choice of the distortion tuple. In the following
theorem, we present a basic result capturing the precise relationship between
$\mathcal{T}$ and $\mathcal{R}$. To the best of our knowledge, this is the
first analytical result that quantifies a tight relationship between utility
and privacy. We briefly sketch the proof here; details can be found in
\cite{LS_VP}.

\begin{theorem}
\label{Lemma_equiv}For a database with a set of utility and privacy metrics,
the tightest utility-privacy tradeoff region $\mathcal{T}$ is the
distortion-equivocation region $\mathcal{R}_{D-E}$.
\end{theorem}

\begin{proof}
The crux of our argument is the fact that for any feasible utility level $D$,
choosing the minimum rate $R\left(  D,E\right)  $, ensures that the least
amount of \textit{information} is revealed about the source via the
reconstructed variables. This in turn ensures that the maximum privacy of the
private attributes is achieved for that utility since, in general, the public
and private variables are correlated. For the same set of utility constraints,
since such a rate requirement is not a part of the utility-privacy model, the
resulting privacy achieved is at most as large as that in $\mathcal{R}_{D-E}$
(see Fig. \ref{Fig_RDE_UP}(a)).
\end{proof}

Implicit in the above argument is the fact that a utility-privacy achieving
code does not perform any better than a rate-distortion-equivocation code in
terms of achieving a lower rate (given by $\log_{2}M/n)$ for the same
distortion and privacy constraints. We can show this by arguing that if such a
code exists then we can always find an equivalent source coding problem for
which the code would violate Shannon's source coding theorem \cite{Shannon_SC}%
. An immediate consequence of this is that a distortion-constrained source
code suffices to preserve a desired level of privacy; in other words,
\textit{the utility constraints require revealing data which in turn comes at
a certain privacy cost that must be borne and vice-versa}. We capture this
observation in Fig. \ref{Fig_RDE_UP}(b) where we contrast existing
privacy-exclusive and utility-exclusive regimes (extreme points of the
utility-privacy tradeoff curve) with our more general approach of determining
the set of feasible utility-privacy tradeoff points.%

\begin{figure*}[tbp] \centering
{\includegraphics[
height=2.8219in,
width=5.4794in
]%
{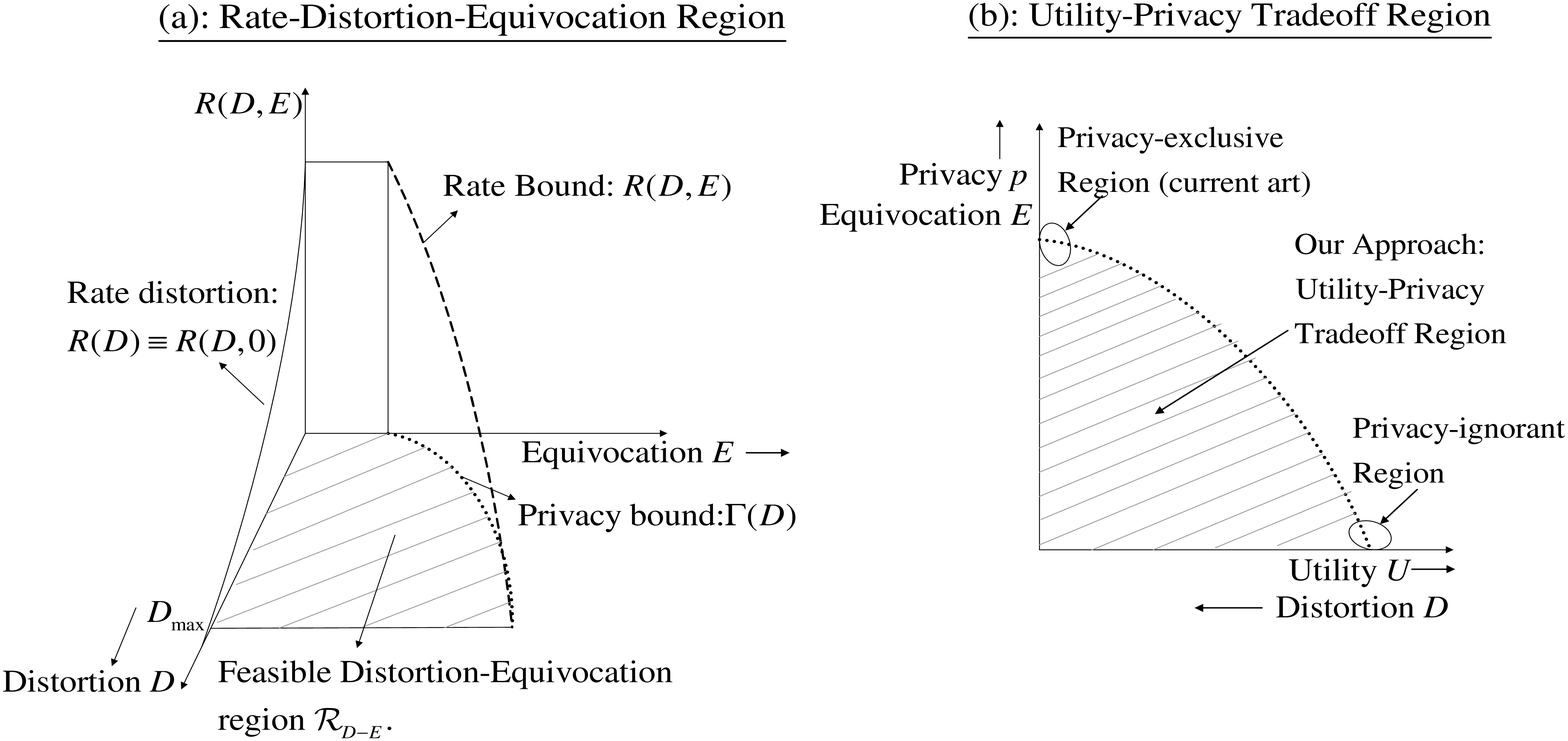}%
}%
\caption{(a) Rate Distortion Equivocation Region; (b) Utility-Privacy Tradeoff Region.}\label{Fig_RDE_UP}%
\end{figure*}%

From an information-theoretic perspective, the power of Theorem
\ref{Lemma_equiv} is that it allows us to study the larger problem of database
utility-privacy tradeoffs in terms of a relatively familiar problem of source
coding with privacy constraints. As noted previously, this problem has been
studied for a specific source model by Yamamoto and here we expand his elegant
analysis to arbitrary database models including those with side information at
the user. Rate for the database can be interpreted as the number of revealed
information bits (precision) per row. Our result shows the tight relationship
between utility, privacy, and precision --\ fixing the value of any one
determines the other two; for example, fixing the utility (distortion $D$)
precisely quantifies the maximal privacy $\Gamma(D)$ and the minimal precision
$R(D,E)$ for any $E$ bounded by $\Gamma(D)$.

\subsubsection{\label{Sec_SI}Capturing the Effects of Side-Information}

It has been illustrated that when a user has access to an external data source
(which is not part of the database under consideration) the level of privacy
that can be guaranteed changes \cite{Sweeney,Dwork_DP}. We cast this problem
in information-theoretic terms as a side information problem.

In an extended version \cite{LS_VP} of this work, we develop the tightest
utility-privacy tradeoff region for the three cases of a) no side information
($L=1$ case studied in \cite{Yamamoto}), b) side information only at the user,
and c) side information at both the source (database) and the user. We present
a result for the case with side information at the user only and for
simplicity, we assume a single utility function, i.e., $L=1$. The proof uses
an auxiliary random variable $U$ along the lines of source coding with side
information \cite{Wyner_Ziv} and bounds the equivocation just as in
\cite[Appendix 1]{Yamamoto}. The following theorem defines the bounds on the
region $\mathcal{R}$ in Definition \ref{Def_RDE} via the functions $\Gamma(D)$
and $R(D,E)$ where $\Gamma(D)$ bounds the maximal achievable privacy and
$R(D,E)$ is the minimal information rate (see Fig. \ref{Fig_RDE_UP}(a)) for
very large databases $\left(  n\rightarrow\infty\right)  $. The proof follows
along the lines of Yamamoto's proof in \cite[Appendix 1]{Yamamoto} and is
skipped in the interest of space.

\begin{theorem}
\label{Pro_Prop1}For a database with side information available only at the
user, the functions $\Gamma(D)$ and $R\left(  D,E\right)  $ and the regions
$\mathcal{R}_{D-E}$ and $\mathcal{R}$ are given by%
\begin{align}
\Gamma\left(  D\right)   &  =\sup_{p\left(  \mathbf{x}_{r},\mathbf{x}%
_{h}\right)  p\left(  u|\mathbf{x}_{r},\mathbf{x}_{h}\right)  \in
\mathcal{P}\left(  D\right)  }H(\mathbf{X}_{h}|UZ)\label{TD_SIu}\\
R\left(  D,E\right)   &  =\inf_{p\left(  \mathbf{x}_{r},\mathbf{x}_{h}\right)
p\left(  u|\mathbf{x}_{r},\mathbf{x}_{h}\right)  \in\mathcal{P}\left(
D,E\right)  }I(\mathbf{X}_{h}\mathbf{X}_{r};U)-I(Z;U) \label{RDE_SIu}%
\end{align}%
\begin{equation}
\mathcal{R}_{D-E}=\left\{  \left(  D,E\right)  :D\geq0,0\leq E\leq
\Gamma\left(  D\right)  \right\}  \label{RDEreg_SIu}%
\end{equation}%
\begin{equation}
\mathcal{R}=\left\{  \left(  R,D,E\right)  :D\geq0,0\leq E\leq\Gamma\left(
D\right)  ,R\geq R\left(  D,E\right)  \right\}  \label{Rreg_SIu}%
\end{equation}
where $\mathcal{P}\left(  D,E\right)  $ is the set of all $p(\mathbf{x}%
_{r},\mathbf{x}_{h},z)p(u|\mathbf{x}_{r},\mathbf{x}_{h})$ such that
$\mathbb{E}\left[  d\left(  \mathbf{X}_{r},g\left(  U,Z\right)  \right)
\right]  \leq D$ and $H(\mathbf{X}_{h}|UZ)\geq E,$ while $\mathcal{P}\left(
D\right)  $ is defined as%
\begin{equation}
\mathcal{P}\left(  D\right)  \equiv%
{\textstyle\bigcup_{H(\mathbf{X}_{h}|\mathbf{X}_{r}Z)\leq E\leq H(\mathbf{X}%
_{h}|Z)}}
\mathcal{P}\left(  D,E\right)  .
\end{equation}

\end{theorem}

While Theorem \ref{Pro_Prop1} applies to a variety of database models, it is
extremely useful in quantifying the utility-privacy tradeoff for the following
special cases of interest.

i) \textit{The single database problem} (i.e., no side information):
\textit{SDB\ is revealed}. Here, we have $Z=0$ and $U=\hat{X}_{r}$, i.e., the
reconstructed vectors seen by the user are the same as the SDB vectors.

ii) \textit{Completely hidden private variables}: \textit{Privacy is
completely a function of the statistical relationship between public, private,
and side information data}. The expression for $R(D,E)$ in (\ref{RDE_SIu})
assumes the most general model of encoding both the private and the public
variables. When the private variables can only be deduced from the revealed
variables, i.e., $\mathbf{X}_{h}-\mathbf{X}_{r}-U$ is a Markov chain, the
expression for $R(D,E)$ in (\ref{RDE_SIu}) will simplify to the Wyner-Ziv
source coding formulation \cite{Wyner_Ziv}, thus clearly demonstrating that
the privacy of the hidden variables is a function of both the correlation
between the hidden and revealed variables and the distortion constraint.

iii)\ \textit{Census and data mining problems without side information}:
\textit{Information rate completely determines the degree of privacy
achievable}. For $Z=0$, setting $\mathbf{X}_{r}=\mathbf{X}_{h}\equiv
\mathbf{X}$ (such that $U=\mathbf{\hat{X}}$), we obtain the census/data mining
problem discussed earlier. In general, due to an additional equivocation
constraint, $R(D,E)\geq R(D)$; however, for this case in which all the
attributes in the database are public, since $\Gamma(D)=H(\mathbf{X}%
)-R(D,E)\leq H(\mathbf{X})-R(D)$, and $R(D)$ is achievable using a
rate-distortion code, the largest possible equivocation is also achievable.
Our analysis thus formalizes the intuition in \cite{Ag_Ag} for using the
mutual information as an estimate of the privacy lost. However in contrast to
\cite{Ag_Ag} in which the underlying perturbation model is an additive noise
model, we assume a perturbation model most appropriate for the input
statistics, i.e., the stochastic relationship between the output and input
variables is chosen to minimize the rate of information transfer. 

\section{\label{Sec_IV}Illustration of Results}

We illustrate our results for two types of databases: one, a
\textit{categorical} database and the other a \textit{numerical} database.
Categorical data are typically discrete data sets comprising of information
such as gender, social security numbers and zipcodes that provide (meaningful)
utility only if they are mapped within their own set. On the other hand,
without loss of generality numeric data can be assumed to belong to the set of
real numbers. In general, a database will have a mixture of categorical and
numerical attributes but for the purpose of illustration, we assume that the
database is of one type or the other, i.e., every attribute is of the same
kind. In both cases, we assume a single utility (distortion) function. We
discuss each example in detail below.

\begin{example}
Consider a categorical database with $K\geq1$ attributes. In general, the
$k^{th}$ attribute $X_{k}$ takes values in a discrete set $\mathcal{X}_{k}$ of
cardinality $M_{k}$. For our example, we model the utility as a single
distortion function of all attributes, and therefore, it suffices to view each
entry (a row of all $K$ attributes) of the database as generated from a single
source $X$ of cardinality $M$, i.e., $X\sim p(x),$ $x\in\left\{
1,2,\ldots,M\right\}  $. For this arbitrary discrete source model, we assume
that the output sample space $\mathcal{\hat{X}=X}$ and consider the
generalized Hamming distortion as the utility function such that the average
distortion $D$ is given by
\begin{equation}
D=E\left[  d(X,\hat{X})\right]  =\Pr\left\{  X\not =\hat{X}\right\}
.\label{CatDB_Dist}%
\end{equation}
For $K=1,$ one can show that $R(D,E)\equiv R(D)$ \cite{LS_VP}; this is because
the maximum achievable equivocation is bounded as $\Gamma(D)=H(X)-R(D,E)\leq
H(X)-R(D)$ with equality when $R(D)$ is achievable. It has been shown by
Erokhin \cite{Erokhin} and Pinkston \cite{Pinkston} that $R(D)$ is achieved by
upside down waterfilling such that
\begin{equation}
p(\hat{x})=\frac{\left(  p(x)-\lambda\right)  ^{+}}{\sum_{x\in\mathcal{X}%
}\left(  p(x)-\lambda\right)  ^{+}}\label{DBCat_pxhat}%
\end{equation}
and the `test channel' is given by%
\begin{equation}
p(x|\hat{x})=\left\{
\begin{array}
[c]{ll}%
\overline{D}, & x=\hat{x}\\
\lambda, & x\not =\hat{x},x\in\mathcal{\hat{X}}_{\text{supp}}\\
p_{k}, & x=k\not \in \mathcal{\hat{X}}_{\text{supp}}%
\end{array}
\right.  \label{DBCat_pxhatx}%
\end{equation}
where $\overline{D}=1-D$, $\lambda$ is chosen such that $\sum_{\hat{x}}%
p(\hat{x})p(x|\hat{x})=p(x)$, $p_{k}=p\left(  x=k\right)  $, and
$\mathcal{\hat{X}}_{\text{supp}}=\left\{  x:p(x)-\lambda>0\right\}  .$ The
maximum achievable equivocation, and hence, the largest utility-privacy
tradeoff region is%
\begin{equation}
\Gamma(D)=-\overline{D}\log\overline{D}-\left\vert \mathcal{\hat{X}%
}_{\text{supp}}\right\vert \lambda\log\lambda-\sum_{k\not \in \mathcal{\hat
{X}}_{\text{supp}}}p_{k}\log p_{k}.
\end{equation}

\end{example}

\begin{remark}
The distortion function chosen in (\ref{CatDB_Dist}) captures the fact that
for categorical data the utility (fidelity) of the revealed data is reduced if
any entry is changed from its original value. The optimal upside down
waterfilling solution in (\ref{DBCat_pxhat}) has the effect of `flattening'
the output distribution, and thus, as in (\ref{DBCat_pxhat}) the source
samples with very high or very low probabilities (relative to the waterfilling
level) are ignored (thereby minimizing the information transfer rate). This in
turn maximizes the privacy achieved since the outliers that are easiest to
infer are eliminated. Eliminating outliers, referred to as information
suppression or aggregation, is the privacy-preserving technique of choice for
the statistics community .
\end{remark}

\begin{example}
In this example we model a numerical database. We consider a $K=2$ database
where both attributes $X$ and $Y$ are jointly Gaussian with zero means and
variances $\sigma_{X}^{2}$ and $\sigma_{Y}^{2}$, respectively, and with
correlation coefficient $\rho=E\left[  XY\right]  /\left(  \sigma_{X}%
\sigma_{Y}\right)  $. This model applies for numeric data such as height and
weight measures which are generally assumed to be normally distributed. We
assume that for every entry only one of the two attributes, say $X$, is
revealed while the other, say $Y$, is hidden such that $Y-X-\hat{X}$ forms a
Markov chain. The rate-distortion-equivocation region for this case can be
obtained directly from Yamamoto's results \cite{Yamamoto} with appropriate
substitution for a jointly Gaussian source. Furthermore, due to the Markov
relationship between of $X,Y,$ and $\hat{X}$, the minimization of $I(X;\hat
{X})$ is strictly over $p(\hat{x}|x),$ and thus, simplifies to the familiar
rate-distortion problem for a Gaussian source $X$ which in turn is achieved by
choosing the reverse channel from $\hat{X}$ to $X$ as an additive white
Gaussian noise channel with variance $D$ (average distortion). The maximal
equivocation achieved thus is
\begin{equation}%
\begin{array}
[c]{cc}%
\Gamma(D)=\sigma_{Y}^{2}\left[  \left(  1-\rho^{2}\right)  +\rho^{2}D\left/
\sigma_{X}^{2}\right.  \right]  , & D\leq\sigma_{X}^{2}.
\end{array}
\end{equation}
Therefore, $\Gamma(D)$ is a minimum for $D=0$ ($X$ revealed perfectly) in
which case only the data independent of $X$ in $Y$ can be private, and is a
maximum equal to the entropy of $Y$ at the maximum distortion $D=\sigma
_{X}^{2}$. Thus, the largest utility-privacy tradeoff region is simply the
region enclosed by $\Gamma(D).$
\end{example}

\section{\label{Sec_V}Concluding Remarks}

We have presented an abstract model for databases with an arbitrary number of
public and private variables, developed application-independent privacy and
utility metrics, and used rate distortion theory to determine the fundamental
utility-privacy tradeoff limits. Future work includes eliminating the row
independence (i.i.d) assumption, modeling and studying tradeoffs for multiple
query databases, and relating current approaches in computer science and our
universal approach. An equally pertinent question is to understand whether our
formalism can be extended to study privacy-utility tradeoffs for less
structured datasets as well as social networks.

\bibliographystyle{IEEEtran}
\bibliography{DB_refs}

\end{document}